\documentclass{article}

\usepackage[utf8]{inputenc}
\usepackage[left=3cm, right=3cm, top=2cm, bottom=2cm]{geometry}

\usepackage[boxed, noend]{algorithm2e}

\usepackage{amsmath, mathtools}
\usepackage{amssymb}
\usepackage{amsthm}
\usepackage{authblk}
\usepackage{comment}
\usepackage{xcolor}
\usepackage[bb=boondox]{mathalfa}
\usepackage{tcolorbox}
\usepackage[noadjust]{cite}
\usepackage{url}

\newtcolorbox{mytheorem}[1][]{
    colback=blue!5!white,    % Background color
    colframe=blue!75!black,  % Frame color
    fonttitle=\bfseries,     % Title font
    coltitle=black,          % Title color
    title={#1}              % Title of the box
}

\newtcolorbox{myalg}[1][]{
    float,
    colback=green!5!white,    % Background color
    colframe=green!75!black,  % Frame color
    fonttitle=\bfseries,     % Title font
    coltitle=black,          % Title color
    title={#1}              % Title of the box
}

\usepackage{mdframed}
\newmdenv[
    backgroundcolor=yellow!10,
    linecolor=yellow!50!black,
    skipabove=10pt,
    skipbelow=10pt,
    leftmargin=10pt,
    rightmargin=10pt,
    innertopmargin=10pt,
    innerbottommargin=10pt,
    roundcorner=5pt,
]{algorithmbox}

\newtheorem{theorem}{Theorem}

\newtheorem{lemma}[theorem]{Lemma}

\newcommand{\eps}{\varepsilon}

\newcommand{\sbs}{\subseteq}

\newcommand{\rr}{\mathbb{R}}

\newcommand{\Ahat}{\widehat{A}}

\newcommand{\nn}{\mathbb{N}}

\newcommand{\R}{\mathbb{R}}
\newcommand{\rank}{\texttt{rank}}
\newcommand{\cutrank}{\texttt{cutrank}}

\SetKwBlock{Loop}{Loop}{EndLoop}

\DeclareMathOperator{\tower}{tower}

\DeclareMathOperator{\disc}{Disc}

\newcommand{\Szemeredi}{Szemer\'{e}di}

\title{Notes on the Linear Algebraic View of Regularity Lemmas\footnote{These notes are a pared-down version of Tuong Le's undergraduate honors thesis at the University of Michigan \cite{thesis}, and therefore contain significant text overlap.}}
\author{Greg Bodwin and Tuong Le\\University of Michigan\\\texttt{\{bodwin, tuongle\}@umich.edu}}
\date{}

\begin{document}

\maketitle

\begin{abstract}
When regularity lemmas were first developed in the 1970s, they were described as results that promise a partition of any graph into a ``small'' number of parts, such that the graph looks ``similar'' to a random graph on its edge subsets going between parts.
Regularity lemmas have been repeatedly refined and reinterpreted in the years since, and the modern perspective is that they can instead be seen as \emph{purely linear-algebraic} results about sketching a large, complicated matrix with a smaller, simpler one.
These matrix sketches then have a nice interpretation about partitions when applied to the adjacency matrix of a graph.

In these notes we will develop regularity lemmas from scratch, under the linear-algebraic perspective, and then use the linear-algebraic versions to derive the familiar graph versions.
We do not assume any prior knowledge of regularity lemmas.
Some comfort with linear algebra will be helpful, although we recap many of the relevant linear-algebraic definitions as we go.
\end{abstract}

\thispagestyle{empty}
\setcounter{page}{0}

\clearpage

\section{Introduction}

A large and active research area within theoretical computer science is that of \emph{graph sparsification}.
Broadly, the goal is to take some large graph $G$, and then compute a much smaller or simpler graph $H$ (a ``sparsifier'') that acts as a low-error representation of $G$ with respect to some particular structural properties.
Some examples of widely-deployed sparsifiers include spanners \cite{PU89, survey}, preservers \cite{CE06, AB24}, flow/cut/spectral sparsifiers \cite{ST11, BSS09}, mimicking networks \cite{KR13, Liu23}, etc.

Some graph sparsifier paradigms are, under the hood, not really about graphs after all: they can instead be viewed as a way to approximate a large \emph{matrix} $A$ with a smaller or simpler matrix $\Ahat$.
This matrix approximation may then be interpreted as a sparsifier when applied to some matrix associated to a graph, like its adjacency matrix or Laplacian.
Spectral sparsifiers are a famous example of this: in modern papers, the theorems, proofs, and constructions of high-quality spectral sparsifiers are usually given in purely linear-algebraic terms.

\emph{Regularity lemmas} are another such example.
They were originally stated and proved using the language of graphs, and sometimes interpreted as sparsifiers, but eventually a linear-algebraic view of things began to develop.
Following the first regularity lemma by \Szemeredi{} in the 1970's \cite{S75}, the matrix perspective was initiated decades later in the seminal work of Frieze and Kannan \cite{FK99}, and then significantly developed by Lov{\'a}sz and Szegedy \cite{LS07}.
From this linear-algebraic standpoint, a regularity lemma is the general term for a result that takes a matrix $A$, and then gives a matrix $\Ahat$ that simultaneously satisfies:
\begin{itemize}
\item (Low-Rankness) $\Ahat$ has low rank, or perhaps it satisfies some reasonable combinatorial analog of low-rankness.

\item (Low-Error Approximation) $A - \Ahat$ has a small top singular value, or perhaps small top-$k$ singular values, or perhaps some reasonable combinatorial analog of these properties.
\end{itemize}

The goal of these notes is to develop regularity lemmas from scratch, with a heavy emphasis on this matrix perspective.
In particular, we start with the classical low-rank approximation theorem, and then we set up and prove a new ``strong'' variant of it, which expresses the key idea that distinguishes the Strong and \Szemeredi{} regularity lemmas \cite{S75, AFKS00} from the weak regularity lemma of Frieze and Kannan \cite{FK99}.
We then use this strong regularity lemma to derive classical graph-theoretic phrasings of regularity lemmas, to clarify its relationship.

\section{Preliminaries and Notation}

The \emph{Frobenius Norm} of a matrix $A$ is defined by
$$\|A\|_F = \left(\sum \limits_{i,j} A_{i,j}^2\right)^{1/2}.$$
Equivalently, it can be defined as the $\ell^2$ norm of the multiset of singular values of $A$; that is,
$$\|A\|_F := \left(\sum \limits_{\sigma \text{ singular value of $A$}} \sigma^2\right)^{1/2}.$$
Associated to the Frobenius Norm is the Frobenius Inner Product over matrices, which functions like the entrywise vector inner product:
$$\langle A, B \rangle_F := \sum \limits_{i, j} A_{ij} B_{ij}.$$
We will frequently reference the ``projection'' of one matrix onto another; formally, this is the projection under this Frobenius inner product, and so it is the same as classical vector projection.
In other words, the projection of a matrix $A$ onto a matrix $B$ is defined as the matrix $\lambda B$, where $\lambda $ is the scalar selected such that
$$\left\langle A - \lambda B, \lambda B\right\rangle_F = 0.$$
It will be useful to discuss the $\ell^2$ mass of just the top few singular values of a matrix.
We could write:
$$\|A\|_{F[k]} := \left(\sum \limits_{\sigma \text{ one of the top $k$ singular values of $A$}} \sigma^2\right)^{1/2},$$
but to foreshadow some other matrix norms that we will come later in these notes, we can equivalently view this norm as:
$$\|A\|_{F[k]} := \max \limits_{A' \text{ is a projection of } A \text{ onto a matrix of rank } \le k} \|A'\|_F.$$
The special case $\|A\|_{F[1]}$ is simply the top singular value of $A$, which is also known as the operator norm or the spectral norm.
We will favor this non-standard $F[1]$ notation for operator/spectral norm in order to emphasize the progression of theorems proved in these notes.

\section{Low Rank Approximation}

\subsection{Weak Low-Rank Approximation}

We will start with the following well-known theorem:

\begin{mytheorem}
\begin{theorem} [Weak Low-Rank Approximation] \label{thm:lra}
For any matrix $A$ and $\eps > 0$, there is a matrix $\Ahat$ satisfying:
\begin{itemize}
\item (simplicity) $\Ahat$ has rank $< \eps^{-2}$, and
\item (approximating) $\|A - \Ahat\|_{F[1]} \le \eps \|A\|_F$.
\end{itemize}
\end{theorem}
\end{mytheorem}

This theorem is deeply related to the singular value decomposition (SVD), but it does not require the full machinery of the SVD to show that this matrix $\Ahat$ exists.
We can use the following simple greedy algorithm:

\DontPrintSemicolon
%\begin{algorithmbox}
%\begin{minipage}{0.6\textwidth}
\begin{algorithm}[h]
\textbf{Input:} matrix $A$, error $\eps > 0$~\\~\\

$\widehat{A} \gets 0$\;

\While{there exists a rank-1 projection $Q$ of $A - \widehat{A}$ with $\|Q\|_F > \eps \|A\|_F$}{
    $\widehat{A} \gets \widehat{A} + Q$\;
}
\Return{$\widehat{A}$}

\caption{\label{alg:vlra} Algorithm for Low-Rank Approximation (Theorem \ref{thm:lra})}
\end{algorithm}
%\end{minipage}
%\end{algorithmbox}

Correctness of Algorithm \ref{alg:vlra} is easy: if and when it halts, it is precisely the definition of the norm $F[1]$ that
$$\|A - \widehat{A}\|_{F[1]} \le \eps \|A\|_F.$$
So it remains to prove the rank bound.
In each round we add a rank-1 matrix to $\Ahat$, which means that we only need to argue that the algorithm always halts within $<\eps^{-2}$ iterations of its main loop.
The idea here is to use $\|A - \Ahat\|_F^2$ as a potential function, which will change each time we update $\Ahat \gets \Ahat+Q$.
Notice that:
\begin{itemize}
\item Initially $\Ahat = 0$, so the potential starts at $\|A\|_F^2$.
\item Potential must remain nonnegative at all times, since it is measured by a norm.
\item Since $Q$ is a projection of $A - \Ahat$, it is orthogonal to $A - \widehat{A} - Q$ (under Frobenius inner product), which means we have the identity 
$$\|A - \Ahat - Q\|_F^2 + \|Q\|_F^2 = \|A - \Ahat\|_F^2.$$
Rearranging, we get
$$\|Q\|_F^2 = \|A - \widehat{A}\|_F^2 - \|A - \widehat{A} - Q\|_F^2.$$
The right-hand side is precisely the amount that potential reduces when we update $\Ahat \gets \Ahat + Q$.
The left-hand side can be lower bounded by $\|Q\|_F^2 > \eps^2 \|A\|_F^2$, due to the condition of the main loop of the algorithm.
So the potential decreases by strictly more than $\eps^2 \|A\|_F^2$ in each round.
\end{itemize}

Putting these three facts together, we conclude that the algorithm halts within $<\eps^{-2}$ rounds.

\subsection{Weak Top-$k$ Low-Rank Approximation}

The statement of Theorem \ref{thm:lra} draws attention to the tradeoff between rank and approximation quality, governed by $\eps$.
This is \textbf{not} actually the tradeoff that is changed in subsequent regularity lemmas.
Instead, the tradeoff of interest is between rank and \emph{number of singular values approximated}.
To highlight this, let us write down the following straightforward corollary of Weak Low-Rank Approximation:
\begin{mytheorem}
\begin{theorem} [Weak Top-$k$ Low-Rank Approximation] \label{thm:klra}
For any matrix $A$, $\eps > 0$, and positive integer $k$, there is a matrix $\Ahat$ satisfying:
\begin{itemize}
\item (simplicity) $\Ahat$ has rank $< k \eps^{-2}$, and
\item (approximating) $\|A - \Ahat\|_{F[k]} \le \eps \|A\|_F$.
\end{itemize}
\end{theorem}
\end{mytheorem}
To prove this, apply Weak Low-Rank Approximation with parameter $\eps/\sqrt{k}$.
Note that all singular values of $A - \widehat{A}$ are between $0$ and $\eps / \sqrt{k}$, and so we have
\begin{align*}
\|A - \Ahat\|_{F[k]} &= \left(\sum \limits_{\sigma \text{ one of the top $k$ singular values of $A - \widehat{A}$}} \sigma^2\right)^{1/2}\\
&\le \left(k \cdot \left(\frac{\eps}{\sqrt{k}} \cdot \|A\|_F\right)^2\right)^{1/2}\\
&=\eps \|A\|_F.
\end{align*}

With this theorem, there is now a tradeoff governed by $k$, in addition to a tradeoff governed by $\eps$.
When we move to our next regularity lemma, the nature of this tradeoff governed by $k$ is the one that will change.

% In a sense, these tradeoffs are both exactly tight, as evidenced by the following simpler lower bound:

% \begin{mytheorem}
% \begin{theorem} [Weak Top-$k$ Low-Rank Approximation, Lower Bound] \label{thm:klralb}
% For any $\eps > 0$ and positive integer $k$, there is a matrix $A$ such that there does \textbf{not} exist a matrix $\Ahat$ satisfying
% \begin{itemize}
% \item (simplicity) $\Ahat$ has rank $k \eps^{-2}$, and
% \item (approximating) $\|A - \Ahat\|_{F[k]} \le o\left(\eps \|A\|_F\right)$.
% \end{itemize}
% \end{theorem}
% \end{mytheorem}
% \begin{proof}
% Let $A$ be the identity matrix of rank $k \eps^{-2} + k$.
% For any matrix $\widehat{A}$ of rank $<k \eps^{-2} - 1$, the codimension of $\widehat{A}$ has rank $k$, and thus
% \begin{align*}
% \|A - \widehat{A}\|_{F[k]} &\ge \left( 1 + \dots + 1\right)^{1/2}\\
%                            &= \sqrt{k}\\
%                            &= \eps \cdot \left(k \eps^{-2} \right)^{1/2}\\
%                            &= \eps \cdot \Omega\left(\|A\|_F\right). \tag*{\qedhere}
% \end{align*}
% \end{proof}

\subsection{Strong Low-Rank Approximation}

Theorem \ref{thm:klra} defines a tradeoff curve, governed by $k$, between the rank of $\Ahat$ and the number of singular values included in the approximation.
Is this tradeoff curve tight?
It turns out that there are two possible senses in which we can investigate tightness of this tradeoff.

Theorem \ref{thm:klra} gives a \emph{universally} attainable tradeoff curve: for any input matrix $A$, we can choose any $k$ that we like, and so we can reach any point on the tradeoff curve that we like.
In that sense, it is optimal: there is no better universally-attainable tradeoff curve.
(For example, a tight lower bound is given by the graph with $\lceil k \eps^{-2}\rceil$ ones along its diagonal and all other entries set to zero; we will not give details on the proof of correctness for this lower bound but it is straightforward.) 

But in some applications, we only need an \emph{existentially} attainable tradeoff curve: for any input matrix $A$, \emph{there exists} at least one choice of $k$ for which the corresponding matrix $\Ahat$ exists.
This is a weaker condition, and so we could dream of it holding for better tradeoff curves.
That is exactly what is achieved by the strong version of low-rank approximation.
For intuition, we note that in this approach, we should intuitively expect the \emph{support} of the tradeoff curve to be a relevant parameter.
That is: if we have a larger pool of choices of $k$, then it should be easier to show that one exists satisfying a theorem.

In the following, given a function $f$ an integer $i \ge 0$, we will denote by $f^{(i)}$ the recursively-defined value:
$$f^{(i)} :=\begin{cases}
0 & \text{if } i=0\\
f^{(i-1)} + f\left(f^{(i-1)}\right) & \text{if } i\ge 1.
\end{cases}$$

\begin{mytheorem}
\begin{theorem} [Strong Low-Rank Approximation] \label{thm:slra}
For any matrix $A, \eps > 0$, and nondecreasing function $f : \nn \to \nn$, \textbf{there exists} an integer $k \le f^{(\eps^{-2})}$
and a matrix $\Ahat$ satisfying:
\begin{itemize}
\item (simplicity) $\Ahat$ has rank $\le k$, and
\item (approximating) $\|A - \Ahat\|_{F[f(k)]} \le \eps \|A\|_F$.
\end{itemize}
\end{theorem}
\end{mytheorem}

In other words, for any tradeoff curve described by a function $f$ whose support includes the interval $[0, f^{(\eps^{-2})}]$, then we can hit its associated tradeoff curve existentially.
To build intuition, let us discuss some possible choices of the function $f$:
\begin{itemize}
\item When we plug in $f(n)=0$, the rank bound is $k \le 0$, and the left-hand side of the approximation guarantee is identically $0$ (since it measures the mass of the top $0$ singular values).
So the theorem is trivial.

\item When we plug in $f(n)=1$, the rank bound is $k \le \eps^{-2}$, and we control $\|A - \Ahat\|_{F[1]}$, so we exactly recover the original low-rank approximation theorem.

\item When we plug in $f(n) = k$, we exactly recover the top-$k$ low-rank approximation theorem in a similar way.
\end{itemize}

But these functions are all constant (independent of $n$), and the real flexibility of the theorem comes from choices of $f$ that depend on $n$. 
The point here is that one can choose $f$ to be an \emph{arbitrarily} fast-growing function, and therefore get an \emph{arbitrarily good} tradeoff between rank and the number of singular values controlled in the approximation, in exchange for blowing up the support size of $f$ and therefore losing a lot of control over the exact value of the rank.
Despite all this additional flexibility, the algorithm and proof for the Strong Low-Rank Approximation Theorem are still not really different from the one used for weak low-rank approximation:

\DontPrintSemicolon
%\begin{algorithmbox}
%\begin{minipage}{0.6\textwidth}
\begin{algorithm}[h]
\textbf{Input:} matrix $A$, error $\eps > 0$~\\~\\

$\widehat{A} \gets 0$\;

\While{there exists a rank-$f(\rank(\Ahat))$ projection $Q$ of $A - \widehat{A}$ with $\|Q\|_F > \eps \|A\|_F$}{
    $\widehat{A} \gets \widehat{A} + Q$\;
}
\Return{$\widehat{A}$}

\caption{\label{alg:slra} Algorithm for Strong Low-Rank Approximation (Theorem \ref{thm:slra})}
\end{algorithm}

Once again, it is immediate by definition of the norm $F[\cdot]$ that once the algorithm halts, it returns a matrix $\widehat{A}$ satisfying the claimed approximation guarantee.
We can also use exactly the same potential argument to conclude that the algorithm runs for $< \eps^{-2}$ rounds.
From there, the rank bound can be proved inductively: initially we have $\rank(\widehat{A})=0$, and then in each round $i \ge 1$, the rank becomes at most
\begin{align*}
&\rank(\widehat{A}) + f(\rank(\widehat{A}))\\
\le \ &f^{(i-1)} + f(f^{(i-1)}) \tag*{inductive hypothesis}\\
= \ &f^{(i)}.
\end{align*}

\section{Matrix-Based Regularity Lemmas}

\subsection{Cut Vectors, Matrices, and Norms}

Regularity lemmas are combinatorial in nature, and so they (usually) use the following discrete analogs of matrix rank.
A \emph{cut vector} is a nonzero vector in which all entries are either $0$ or $1$.
A \emph{cut matrix} is a special kind of rank-1 matrix that is a scalar multiple of an outer product of two cut vectors; in other words, it is a matrix that is constant on some $S \times T$ block and $0$ elsewhere.
The \emph{cutrank} of a matrix $A$ is the smallest integer $k$ for which it is possible to write
$$A = C_1 + \dots + C_k,$$
where each $C_i$ is a cut matrix.
In general, for an $m \times n$ matrix $A$, we have
$$0 \le \rank(A) \le \cutrank(A) \le mn.$$
Analogizing the Frobenius norm and its restrictions, we define:
$$\|A\|_{\blacksquare[k]} := \max \limits_{A' \text{ is a projection of $A$ onto a matrix of cutrank $\le k$}} \| A' \|_F.$$

We remark that there is nothing \emph{too} special about cut vectors in the following content: one could define analogous concepts with respect to any set of vectors $V$ (here $V$ is the set of cut vectors), and even obtain regularity lemmas analogous to the ones in this section with respect to any set of vectors $V$.
We focus on cut vectors only because they give the versions of matrix regularity lemmas that imply the classical graph-theoretic regularity lemmas, as demonstrated in the next section.

\subsection{Remark: Normalized Cut Norm vs.\ Cut Norm}

Most other resources on regularity lemmas will focus on the \emph{cut norm} $\|A\|_{\square}$ rather than our \emph{normalized cut norm} $\blacksquare[\cdot]$; this is the main design choice that distinguishes our approach from others you might find in the wild.
The cut norm (which we will not use at all) is defined as
$$\left\|A\right\|_{\square} := \max \limits_{v, w \text{ cut vectors}} \left| v^T A w \right|.$$
This is \textbf{not} the same as $\|A\|_{\blacksquare[1]}$, which is equivalent to a normalized version, $$\left\|A\right\|_{\blacksquare[1]} := \max \limits_{v, w \text{ cut vectors}} \left| \frac{v^T}{\|v\|_2} A \frac{w}{\|w\|_2} \right|.$$

There are a few reasons that we prefer the normalized cutnorm over the classical cutnorm.
One is that it is easier and more natural to define its higher-order version ($\blacksquare[k]$ rather than $\blacksquare[1]$).
Another is that it connects more closely with the classical low-rank approximation theorem.
Another is that it is a bit stronger: when $A \in \rr^{m \times n}$, we have
$$\|A\|_{\square} \le \sqrt{mn} \cdot \|A\|_{\blacksquare[1]}$$
and through this inequality, the controls we will prove on $\|A\|_{\blacksquare[\cdot]}$ imply the controls we would get on $\|A\|_{\square}$ if we centered that instead.
Finally, the norm $\blacksquare[\cdot]$ is arguably preferable to $\square$ from an algorithmic standpoint.
A famous algorithm by Alon and Naor \cite{AN04} achieves a constant-factor approximation to the cut norm $\square$ in polynomial time, which is good enough to prove regularity lemmas in the perspective that centers $\square$.
However, there is a simpler algorithm that \emph{exactly} computes $\blacksquare[1]$ in polynomial time (\cite{BV22}, c.f.\ Appendix A; this algorithm is based heavily on a previous one by Charikar \cite{Charikar00}).
It is an interesting open problem to compute $\blacksquare[k]$ in polynomial time for general $k$, although with a bit more effort than we will show in the following sections, it is possible to use only the algorithm for $\blacksquare[1]$ to get polynomial-time versions of our regularity lemmas \cite{BV22}.

\subsection{Regularity Lemmas}

The following theorem is known as the \emph{weak regularity lemma}, and was initially proved by Frieze and Kannan \cite{FK99}:

\begin{mytheorem}
\begin{theorem} [Weak Regularity Lemma \cite{FK99}] \label{thm:wlcra}
For any matrix $A$, $\eps > 0$, there is a matrix $\Ahat$ satisfying:
\begin{itemize}
\item (simplicity) $\Ahat$ has cutrank $< \eps^{-2}$, and
\item (approximating) $\|A - \Ahat\|_{\blacksquare[1]} \le \eps \|A\|_F$.
\end{itemize}
\end{theorem}
\end{mytheorem}

We can also have a top-$k$ version of this algorithm, which again follows as a corollary by choosing parameter $\eps / \sqrt{k}$:

\begin{mytheorem}
\begin{theorem} [Weak Top-$k$ Regularity Lemma]
For any matrix $A$, $\eps > 0$, and integer $k$, there is a matrix $\Ahat$ satisfying:
\begin{itemize}
\item (simplicity) $\Ahat$ has cutrank $< k\eps^{-2}$, and
\item (approximating) $\|A - \Ahat\|_{\blacksquare[k]} \le \eps \|A\|_F$.
\end{itemize}
\end{theorem}
\end{mytheorem}

We can also extend this into a strong version in the same way as before:

\begin{mytheorem}
\begin{theorem} [Strong Regularity Lemma \cite{S75, AFKS00}] \label{thm:slcra}
For any matrix $A, \eps > 0$, and function $f : \nn \to \nn$, \textbf{there exists} an integer $k \le f^{(\eps^{-2})}$
and a matrix $\Ahat$ satisfying:
\begin{itemize}
\item (simplicity) $\Ahat$ has cutrank $\le k$, and
\item (approximating) $\|A - \Ahat\|_{\blacksquare[f(k)]} \le \eps \|A\|_F$.
\end{itemize}
\end{theorem}
\end{mytheorem}

There is essentially no difference between the proofs of these theorems and the proofs of the low-rank approximation theorems from the previous section.
The construction algorithm is the same, but with cutrank in place of rank:

\DontPrintSemicolon
%\begin{algorithmbox}
%\begin{minipage}{0.6\textwidth}
\begin{algorithm}[h]
\textbf{Input:} matrix $A$, error $\eps > 0$~\\~\\

$\widehat{A} \gets 0$\;

\While{there exists a cutrank-$f(\cutrank(\widehat{A}))$ projection $Q$ of $A - \widehat{A}$ with $\|Q\|_F > \eps \|A\|_F$}{
    $\widehat{A} \gets \widehat{A} + Q$\;
}
\Return{$\widehat{A}$}

\caption{\label{alg:svlra} Algorithm for Strong Regularity Lemma (Theorem \ref{thm:slcra})}
\end{algorithm}

The proof is also the same as before; nothing significant changes when we use cutrank in place of rank.
We again use $\|A - \Ahat\|_F^2$ as a potential function to argue that the algorithm repeats for only $\eps^{-2}$ many rounds, and again, we have that the cutrank of $\Ahat$ is initially $0$ and then becomes at most $\cutrank(\Ahat) +f(\cutrank(\Ahat))$ in each round, leading to the claimed bound on the final cutrank.
The weak regularity lemma is the special case $f(n) = 1$, and the weak top-$k$ regularity lemma is the special case $f(n) = k$.

\section{Graph-Based Regularity Lemmas \label{sec:graphs}}

Our next task will be to derive the common graph-theoretic phrasings of regularity lemmas as corollaries of the linear-algebraic ones above.
The proofs are basically just applications of Theorem \ref{thm:slcra} plus some algebraic rearrangement, with no major new technical ideas.
In some sense the \emph{point} of these notes is to encourage you to think of Theorem \ref{thm:slcra} as the ``real'' regularity lemma, and the following theorems in this section as downstream graph-theoretic applications.

\subsection{Common Refinements}

In the following, we will discuss the \emph{common refinement} of matrices $\Ahat \in \rr^{n \times n}$, indexed by a vertex set $V$.
This is the partition $V = V_1 \cup \dots \cup V_k$ defined as follows.
Let
$$\Ahat = C_1 + \dots + C_{k'}$$
be a minimal decomposition into cut matrices (so $\cutrank(\Ahat) = k'$).
Each of these cut matrices $C_i$ is nonzero on an $S_i \times T_i$ block, and these sets $\{S_i\} \cup \{T_i\}$ give $2k$ separate bipartitions of $V$.
The common refinement is defined as the maximal partition that refines each of these bipartitions; that is, it puts two nodes $u, v \in V$ into the same part iff they lie on the same side of all $2k$ bipartitions.
The size of a common refinement is bounded by
$$k \le 2^{2k} = 4^k,$$
since each of the $2k$ sets can at most double the number of parts in the common refinement.
Also note that $\Ahat$ is constant on each $V_i \times V_j$ block, since by construction each cut matrix $C_i$ contains either all or none of $V_i \times V_j$ in its associated $S_i \times T_i$ block.

\subsection{A Useful Property of the Frobenius Inner Product \label{sec:frobswitch}}

Several times in the following, we will invoke a technical property of the Frobenius norm/inner product, which we recap here without proof.
Let $s, t \in \rr^n$ and let $A \in \rr^{n \times n}$.
Then we have the identity
\begin{align*}
\frac{s^T}{\|s\|_2} A \frac{t}{\|t\|_2} &= \left\langle \frac{s t^T}{\|s\|_2 \|t\|_2}, A \right\rangle_F.
\end{align*}
Additionally, we have
$$\left\| \frac{s t^T}{\|s\|_2 \|t\|_2} \right\|_F = 1,$$
and therefore the right-hand side of the identity measures the magnitude (in Frobenius norm) of the projection of $A$ onto the matrix $st^T$.
% In particular, that means we have
% $$\frac{s^T}{\|s\|_2} A \frac{t}{\|t\|_2} \le \|A\|_{F[1]},$$
% and in the case where $s, t$ are cut vectors, the matrix $st^T$ has cutrank $1$ and so we have
% $$\frac{s^T}{\|s\|_2} A \frac{t}{\|t\|_2} \le \|A\|_{\blacksquare[1]}.$$

\subsection{The Weak Regularity Lemma}

The weak regularity lemma was originally proved by Frieze and Kannan \cite{FK99}.
For intuition on the statement, imagine that we have an $n$-node graph $G$, and we want to compress it into small space using two steps:
\begin{itemize}
\item Compute a partition of its vertices into a small number of parts,
\item Record a single real number $c_{ij}$ for each pair of parts $(V_i, V_j)$ in the partition, which roughly represents a claim that about a $c_{ij}$ fraction of the possible edges in $V_i \times V_j$ are present in $G$.
\end{itemize}
The goal of the compression is, for arbitrary vertex subsets $S, T$, to be able to accurately estimate the number of edges $e(S, T)$ going between these subsets in $G$.
More specifically, the estimate $\widehat{e(S, T)}$ associated to the compression is computed by checking the sizes of the intersections of $S, T$ with each pair of vertex subsets $V_i, V_j$, and then estimating that a $c_{ij}$ fraction of the edges between these intersections $S \cap V_i, T \cap V_j$ are present in $G$.
The weak regularity lemma guarantees a compression that will achieve additive error $\eps n^2$.

\begin{mytheorem}
\begin{theorem} [Weak Regularity Lemma, Graph Phrasing \cite{FK99}] \label{thm:weakreggraph}
For any $n$-node graph $G = (V, E)$ and any $\eps > 0$, there exists a partition
$$V = V_1 \cup \dots \cup V_k$$
of size $k \le 4^{\eps^{-2}}$, and a real number $c_{ij}$ for each $1 \le i, j \le k$, such that the following holds for any $S, T \subseteq V$.
Letting
$$\widehat{e(S, T)} := \sum \limits_{i, j=1}^k c_{ij} \left|V_i \cap S\right|\left|V_j \cap T\right|,$$
we have
$$\left| e(S, T) - \widehat{e(S, T)}\right| < \eps n^2.$$ 
\end{theorem}
\end{mytheorem}

To prove this, let $A$ be the adjacency matrix of $G$, and apply Theorem \ref{thm:wlcra} (or equivalently, Theorem \ref{thm:slcra} with function $f(n)=1$), to generate an approximating matrix $\Ahat$.
Let $V = V_1 \cup \dots \cup V_k$ be the common refinement of $\Ahat$, which thus has size
$$k \le 4^{\cutrank(\Ahat)} \le 4^{\eps^{-2}}.$$
Since the $\Ahat$ is constant on each $V_i \times V_j$ block, we may let $c_{ij}$ be the common value of all entries in $\Ahat[V_i \times V_j]$.

Now let $S, T \subseteq V$ be arbitrary vertex subsets, and let $s, t \in \rr^n$ be their respective binary indicator vectors.
We then have
\begin{align*}
\left| e(S, T) - \widehat{e(S, T)} \right|&= \left| s^T A t - s^T \Ahat t \right|\\
&= \underbrace{\|s\|_2 \|t\|_2}_{\text{left term}} \cdot \underbrace{\left| \frac{s^T}{\|s\|_2} \left(A - \Ahat \right) \frac{t}{\|t\|_2} \right|}_{\text{right term}}.
\end{align*}

In order to continue simplifying, we inspect the two terms in this product in turn.
For the left term, since $s, t$ are binary vectors in $\rr^n$, they each have norm at most $\sqrt{n}$, and so $\|s\|_2 \|t\|_2 \le n$.
For the right term, by the identity in Section \ref{sec:frobswitch}, this is magnitude of the projection of $A - \widehat{A}$ onto the matrix $st^T$.
Since $st^T$ is a cut matrix, this is by definition at most $\|A - \Ahat\|_{\blacksquare[1]}$.
So we continue:
\begin{align*}
\left| e(S, T) - \widehat{e(S, T)} \right|&\le \underbrace{n}_{\text{bound on left term}} \cdot \underbrace{\left\|A - \Ahat\right\|_{\blacksquare[1]}}_{\text{bound on right term}}\\
&< n \cdot \left(\eps \left\| A \right\|_F\right) \tag*{by Theorem \ref{thm:wlcra}}\\
&\le n \cdot \left(\eps n\right) \tag*{since $A$ is a binary $n \times n$ matrix}\\
&= \eps n^2.
\end{align*}

\subsection{Discrepancy}

Graph versions of regularity lemmas are often phrased using the concept of \emph{discrepancy}, defined as follows.
Given vertex subsets $V_i, V_j$ in a graph $G = (V, E)$, their \emph{discrepancy} is defined as the quantity
$$ \disc(V_i, V_j) := \min_{c_{ij}} \max \limits_{S \subseteq V_i, T \subseteq V_j} \left| e(S, T) - c_{ij} |S||T| \right|,$$
where $e(S, T)$ counts the number of edges with one endpoint in $S$ and the other in $T$.
The intuition for this definition is similar to that for the weak regularity lemma.
Imagine that we try to compress the part of the graph between $V_i$ and $V_j$ down to just a single real number $c_{ij}$, roughly representing a claim that for any node subsets $S \subseteq V_i, T \subseteq V_j$, about a a $c_{ij}$ fraction of the possible $S$-$T$ edges are present in $G$.
The discrepancy is a measure of the accuracy of that estimate; namely, it is the largest number of edges on which this estimate ever differs from the correct cut value.
One would expect the discrepancy of a random graph of density $c_{ij}$ to be fairly small, and therefore discrepancy is sometimes considered a measure of the ``pseudorandomness'' of $G$ between $V_i, V_j$. 

Discrepancy can be related to the $\blacksquare$ matrix norms using the following technical lemma.
The proof is a bit dense, but it does not contain a major new concept; it boils down to unpacking the definitions and applying Cauchy-Schwarz.

\begin{lemma} [Relationship Between Discrepancy and $\blacksquare$ Norm] \label{lem:discsquare}
Let $G = (V, E)$ be an $n$-node graph with adjacency matrix $A \in \rr^{n \times n}$.
Let $\{V_1, \dots, V_k\}$ be disjoint vertex subsets and let $\Ahat \in \rr^{n \times n}$ be a matrix that is constant on its $V_i \times V_j$ block for any two of these subsets $V_i, V_j$.
Then
$$\sum \limits_{i,j=1}^k \disc(V_i, V_j) \le n \cdot \left\| A - \widehat{A} \right\|_{\blacksquare[k^2]}.$$
\end{lemma}
\begin{proof}
For any $i, j$, let $c_{ij}$ denote the value of all entries in $\Ahat[V_i, V_j]$.
Also let $S_{ij} \subseteq V_i, T_{ij} \subseteq V_j$ be the choice of vertex subsets maximizing
$$\left| e(S, T) - c_{ij} |S||T| \right|$$
over all $S \subseteq V_i, T \subseteq V_j$, and let $s_{ij}, t_{ij} \in \rr^n$ be the binary indicator vectors for $S_{ij}, T_{ij}$ respectively.
We can now bound the sum of discrepancies using the following algebra:
\begin{align*}
\sum \limits_{i, j=1}^k \disc(V_i, V_j) &\le \sum \limits_{i, j=1}^k \left| e(S_{ij}, T_{ij}) - c_{ij} \left|S_{ij}\right|\left|T_{ij}\right| \right| \tag*{definition of $\disc$}\\
&\le \sum \limits_{i, j=1}^k \left| s_{ij}^T A t_{ij} - s_{ij}^T \widehat{A} t_{ij} \right| \\
&= \sum \limits_{i, j=1}^k \left( \|s_{ij}\|_2 \|t_{ij}\|_2 \right) \cdot \left| \frac{s_{ij}^T}{\|s_{ij}\|_2}  \left(A - \widehat{A}\right) \frac{t_{ij}}{\|t_{ij}\|_2} \right| \\
&\le \underbrace{\left( \sum \limits_{i, j=1}^k \|s_{ij}\|_2^2 \|t_{ij}\|_2^2 \right)^{1/2}}_{\text{left term}} \cdot \underbrace{\left( \sum \limits_{i, j=1}^k \left( \frac{s_{ij}^T}{\|s_{ij}\|_2} \left(A - \widehat{A}\right) \frac{t_{ij}}{\|t_{ij}\|_2} \right)^2\right)^{1/2}}_{\text{right term}}
%\\ &= \left( n^2 \right)^{1/2} \cdot \left( \sum \limits_{i=1}^k \sum \limits_{j=1}^k \left\langle \frac{s_{ij} t_{ij}^T}{\|s_{ij}\|_2 \|t_{ij}\|_2}, A - \widehat{A} \right\rangle^2\right)^{1/2}.
\end{align*}
where the last step is by the Cauchy-Schwarz inequality.
In order to continue simplifying, we will inspect the two terms in this product in turn.

\paragraph{Left Term.}
On the left: the terms $\|s_{ij}\|_2^2, \|t_{ij}\|_2^2$ are simply $|S_{ij}|$ and $|T_{ij}|$ respectively, which are no larger than their respective host sets $|V_i|$ and $|V_j|$.
So we may simplify
\begin{align*}
\left( \sum \limits_{i, j=1}^k \|s_{ij}\|_2^2 \|t_{ij}\|_2^2 \right)^{1/2} &\le \left( \sum \limits_{i, j=1}^k |V_i||V_j| \right)^{1/2}\\
&= (n^2)^{1/2} \tag*{since $\{V_i\}$ are disjoint and $|V|=n$}\\
&= n.
\end{align*}

\paragraph{Right Term.}
By the discussion in Section \ref{sec:frobswitch}, the expression
\begin{align*}
\frac{s_{ij}^T}{\|s_{ij}\|_2} \left(A - \widehat{A}\right) \frac{t_{ij}}{\|t_{ij}\|_2}
\end{align*}
is the $\ell^2$ norm of the magnitude of the projections of $A - \widehat{A}$ onto the set of matrices
$$\left\{\frac{s_{ij} t_{ij}^T}{\|s_{ij}\|_2 \|t_{ij}\|_2}\right\}.$$
These are in fact orthonormal cut matrices; that is, they are pairwise orthogonal under Frobenius inner product, since their nonzero entries are disjoint.
So the $\ell^2$ norm of the projections of $A - \widehat{A}$ is in fact the magnitude of the projection of $A - \widehat{A}$ onto
$$\texttt{span}\left\{\frac{s_{ij} t_{ij}^T}{\|s_{ij}\|_2 \|t_{ij}\|_2}\right\}.$$
Since there are $k^2$ orthonormal cut matrices in this set, by definition all matrices in the span have cutrank at most $k^2$, and so this term is at most
$\left\| A - \widehat{A} \right\|_{\blacksquare[k^2]}$.

\paragraph{Completing the Calculation.}

Putting the two terms together, we have
\begin{align*}
\sum \limits_{i, j=1}^k \disc(V_i, V_j) &\le \underbrace{n}_{\text{bound on left term}} \cdot \underbrace{\left\| A - \widehat{A} \right\|_{\blacksquare[k^2]}}_{\text{bound on right term}},
\end{align*}
as claimed.

\end{proof}

\subsection{\Szemeredi{} Regularity Lemma, Discrepancy Version}

We will next prove the \Szemeredi{} regularity lemma, unsurprisingly originally proved by \Szemeredi{} \cite{S75}.
There are at least two popular phrasings; the following is the simpler of the two to state and prove, but has slightly weaker guarantees.
Recall that
$$\tower(k) = \underbrace{2^{2^{2^{\dots ^2}}}}_{k \text{ times}}.$$

\begin{mytheorem}
\begin{theorem} [\Szemeredi{} Regularity Lemma, Graph Phrasing, Discrepancy Version] \label{thm:regdisc}
For any $n$-node graph $G = (V, E)$ and any $\eps > 0$, there exists an integer $k$ bounded by
$$k \le \tower(O(\eps^{-2}))$$
and a partition $V = V_1 \cup \dots \cup V_k$ such that
$$\sum \limits_{i, j=1}^k \disc(V_i, V_j) \le \eps n^2.$$ 
\end{theorem}
\end{mytheorem}

To prove this, let $A$ be the adjacency matrix of $G$, and generate an approximating matrix $\widehat{A}$ by applying Theorem \ref{thm:slcra} (the Strong Regularity Lemma) to $A$ with function $f(n) = 16^n$.
Let $k'$ be the integer given by Theorem \ref{thm:slcra} (which is different from $k$ in this theorem).
We can bound this inductively: we can use the crude upper bound
\begin{align*}
f^{(i)} &= f^{(i-1)} + f(f^{(i-1)})\\
&\le 2\cdot f(f^{(i-1)})\\
&\le 32^{f^{(i-1)}},
\end{align*}
and thus $f^{(i)} \le \tower(O(i))$, and so
$$k' \le f^{(\eps^{-2})} \le \tower(O(\eps^{-2})).$$
We define the partition $V = V_1 \cup \dots \cup V_k$ by taking the common refinement of $\widehat{A}$, which thus has size
\begin{align*}
k &\le 2^{2 \cdot \cutrank(\widehat{A})}\\
&\le 4^{k'}
\end{align*}
which is still at most $\tower(O(\eps^{-2}))$.
We may then bound:
\begin{align*}
\sum \limits_{i,j=1}^k \disc(V_i, V_j) &\le n \cdot \left\| A - \Ahat \right\|_{\blacksquare[k^2]} \tag*{by Lemma \ref{lem:discsquare}}\\
&\le n \cdot \left\| A - \Ahat \right\|_{\blacksquare[(4^{k'})^2]}\\
&\le n \cdot (\eps \|A\|_F) \tag*{by Theorem \ref{thm:slcra}}\\
&\le n \cdot (\eps n) \tag*{since $A$ is a binary $n \times n$ matrix}\\
&= \eps n^2,
\end{align*}
completing the proof.

The following closely related version of the \Szemeredi{} regularity lemma can be obtained as a simple corollary.
Let us say that a pair of vertex subsets $(V_i, V_j)$ is \emph{irregular} if we have
$$\disc(V_i, V_j) > \eps |V_i||V_j|.$$

\begin{mytheorem}
\begin{theorem} [\Szemeredi{} Regularity Lemma, Graph Phrasing, Irregularity Version] \label{thm:graphregbd}
For any $n$-node graph $G = (V, E)$ and any $\eps > 0$, there exists an integer $k$ bounded by
$$k \le \tower(O(\eps^{-4}))$$
and a partition $V = V_1 \cup \dots \cup V_k$ such that
$$\sum \limits_{(V_i, V_j) \text{ irregular}} |V_i||V_j| < \eps n^2.$$
\end{theorem}
\end{mytheorem}

To prove this, apply Theorem \ref{thm:regdisc} with parameter $\eps^2$, giving a partition $V = V_1 \cup \dots \cup V_k$ with $k \le \tower(O(\eps^{-4}))$, satisfying
$$\sum \limits_{i, j=1}^k \disc(V_i, V_j) \le \eps^2 n^2.$$

So we have:
\begin{align*}
\sum \limits_{(V_i, V_j) \text{ irregular}} |V_i||V_j| &= \sum \limits_{(V_i, V_j) \ \mid \ \disc(V_i, V_j) > \eps |V_i||V_j|} |V_i||V_j|\\
&< \eps^{-1} \cdot \sum \limits_{(V_i, V_j) \ \mid \ \disc(V_i, V_j) > \eps |V_i||V_j|} \disc(V_i, V_j)\\
&\le \eps^{-1} \cdot \sum \limits_{(V_i, V_j)} \disc(V_i, V_j)\\
&\le \eps^{-1} \cdot \eps^2 n^2 \tag*{Theorem \ref{thm:regdisc}}\\
&= \eps n^2.
\end{align*}

\subsection{\Szemeredi{} Regularity Lemma, Exceptional Set Version}

Yet another popular phrasing of the regularity lemma hides some nodes in an ``exceptional set,'' while enforcing stronger properties on the remaining parts: all other parts have the same size, and most of them are regular.

\begin{mytheorem}
\begin{theorem} [\Szemeredi{} Regularity Lemma, Graph Phrasing, Exceptional Set Version]
For any $n$-node graph $G = (V, E)$ and any $\eps > 0$, there exists an integer $k$ bounded by
$$k \le \tower(O(\eps^{-4}))$$
and a partition $V = V_0 \cup V_1 \cup \dots \cup V_k$ such that:
\begin{itemize}
\item $|V_0| < \eps n$ ($V_0$ is called the ``exceptional set'')
\item For all $1 \le i, j \le k$, we have $|V_i| = |V_j|$
\item At most $\eps k^2$ pairs of parts $(V_i, V_j)$ are irregular.
\end{itemize}
\end{theorem}
\end{mytheorem}

To prove this, let $A$ be the adjacency matrix of $G$, and generate an approximating matrix $\Ahat$ by applying Theorem \ref{thm:slcra} (the Strong Regularity Lemma) to $A$ with error parameter $\eps^2$ and function
$$f(n) = \eps^{-2} \cdot 16^{n}.$$
Let $k'$ be the associated integer; following a similar calculation to before, we can bound
\begin{align*}
k' &\le f^{(\eps^{-4})} \le \tower(O(\eps^{-4})).
\end{align*}
(We note that the additional $\eps^{-2}$ factor in $f(n)$ will be needed later, but it does not affect the quantitative bounds much; its entire contribution to the value of $f^{(\eps^{-4})}$ is dominated by extending the tower height by, say, $+O(\log \eps^{-1})$ additional levels.)

Define a partition $V = V_1 \cup \dots \cup V_{k''}$ by taking the common refinement of the cut matrices used to build $\widehat{A}$ (this time, the common refinement is still not yet the partition that satisfies the theorem).
Then, for each part $V_i$ in the common refinement, take a subpartition
$$V_i = V_i^0 \cup V_i^1 \cup \dots \cup V_i^{a}$$
for which the parts in this subpartition have sizes
$$\left|V_i^0\right| < \left|V_i^1\right| = \dots = \left|V_i^a\right| = \frac{\eps n}{k''}$$
($V_i^0$ is possibly empty, and we assume for convenience that $k''$ divides $\eps n$ but otherwise terms may be rounded without significantly affecting the following argument).
Move all the nodes in $V_i^0$ into the exceptional set $V_0$.
The total size of the exceptional set becomes
\begin{align*}
\left|V_0\right| &= \sum \limits_{i=1}^{k''} \left|V_i^0\right|\\
&< \sum \limits_{i=1}^{k''} \frac{\eps n}{k''}\\
&= \eps n
\end{align*}
as claimed.
Now let $k$ be the total number of remaining non-exceptional parts across all subpartitions; we claim that these parts satisfy the rest of the theorem.
By construction they all have the same size $\eps n / k''$, and hence there are $k \le  k'' / \eps$ of them.
We can bound
\begin{align*}
k &\le \frac{k''}{\eps}\\
&\le \eps^{-1} \cdot 4^{k'}\\
&\le \eps^{-1} \cdot 4^{\tower(O(\eps^{-4}))}\\
&\le \tower(O(\eps^{-4})).
\end{align*}
So now it remains to count the number of non-exceptional irregular pairs of parts.
We have
\begin{align*}
\left|\left\{ \left(V_i^a, V_j^b\right) \text{ irregular} \right\}\right| \cdot \underbrace{\eps \cdot \left(\frac{n}{k}\right)^2}_{\substack{\text{lower bound on}\\\text{contribution to discrepancy}\\\text{of each irregular pair}}} &< \sum \limits_{i, j, a, b} \disc\left(V_i^a, V_j^b\right)\\
&\le n \cdot \left\|A - \widehat{A}\right\|_{\blacksquare[k^2]} \tag*{by Lemma \ref{lem:discsquare}}\\
&\le n \cdot \left\|A - \widehat{A}\right\|_{\blacksquare[(\eps^{-1} k'')^2]}\\
&\le n \cdot \left\|A - \widehat{A}\right\|_{\blacksquare[(\eps^{-1} 4^{k'})^2]}\\
&\le n \cdot \left\|A - \widehat{A}\right\|_{\blacksquare[\eps^{-2} 16^{k'}]}\\
&\le n \cdot (\eps^2 \|A\|_F) \tag*{by Theorem \ref{thm:slcra}}\\
&\le n \cdot (\eps^2 n) \tag*{since $A$ is a binary $n \times n$ matrix}\\
&= \eps^2 n^2.
\end{align*}
Thus, comparing the first to the last term and rearranging, we have
\begin{align*}
\left|\left\{ \left(V_i^a, V_j^b\right) \text{ irregular} \right\}\right| &< \eps k^2,
\end{align*}
completing the proof.

\section{Other Directions}

We will conclude by mentioning some technical directions that were not covered explicitly in these notes.

\paragraph{``Large Subset'' Regularity Lemmas.}

Some phrasings of regularity lemma will use a slightly different concept of irregularity than ours.
The alternative is that a pair of parts $(V_i, V_j)$ are ``irregular'' if there exist \textbf{large} subsets $S \subseteq V_i, T \subseteq V_j$, specifically of sizes $|S| \ge \eps |V_i|, T \ge \eps |V_j|$, with
$$\left| \frac{e(V_i V_j)}{|V_i||V_j|} - \frac{e(S, T)}{|S||T|} \right| \le \eps,$$
in other words, the edge density between $S, T$ fluctuates significantly from the edge density between $V_i, V_j$.
This is essentially just a cosmetic difference.
With a bit of algebra, one can show that this definition of irregularity is equivalent to ours, up to polynomial changes in $\eps$ that do not ultimately affect the statements of regularity lemmas.

\paragraph{``Edge Density'' Regularity Lemmas.}

Some phrasings of regularity lemmas use a slightly different definition of discrepancy, where the parameter $c_{ij}$ is not set by a $\min$ as in our definition, but rather is some fixed statistic of the parts $V_i, V_j$.
For example, a common one is
$$ \disc(V_i, V_j) := \max \limits_{S \subseteq V_i, T \subseteq V_j} \left| e(S, T) - \underbrace{\left( \frac{e(V_i, V_j)}{|V_i||V_j|}\right)}_{c_{ij}} |S||T| \right|,$$
i.e., $c_{ij}$ is the edge density between $V_i$ and $V_j$.
Although we are not aware of any application of regularity lemmas where this fixing of $c_{ij}$ is really necessary, if it is desired, it requires a tweak to the \emph{algorithm} rather than just a change in the analysis.
We would essentially have to push the part of the analysis where we take a common refinement of $\Ahat$ into \emph{each round} of the algorithm, rather than doing it once at the very end as in our approach.
Letting $V_1 \cup \dots \cup V_k$ be the common refinement of $\Ahat$ at a point in the algorithm, we would define $\Ahat$ by projecting $A$ onto the span of the refined cut matrices associated to each $V_i \times V_j$ block, and a bit of matrix algebra shows that the entries assigned to this block will indeed be the edge density of $A$ on that block.

\paragraph{The Graph-Theoretic Strong Regularity Lemma.} We have not covered the graph-theoretic version of the strong regularity lemma by Alon, Fischer, Krivelevich, and Szegedy \cite{AFKS00}.
Roughly speaking, this lemma promises both an \emph{outer} partition
$V = V_1 \cup \dots \cup V_k,$
and also an \emph{inner} partition that refines the outer one
$$V = \left(V_1^1 \cup \dots \cup V_1^{j_1}\right) \cup \dots \cup \left(V_k^1 \cup \dots \cup V_k^{j_k}\right).$$
The necessary properties are:
\begin{itemize}
\item The inner partition satisfies the regularity lemma with respect to a parameter $\eps_1$, and
\item For the ``average'' pair of outer parts $V_i, V_j$, the inner refinement of these parts satisfies the regularity lemma on the bipartite subgraph between these parts, with respect to a parameter $\eps_2$. 
\end{itemize}

This regularity lemma and its various phrasings can again be recovered from Theorem \ref{thm:slcra} using similar methods to the \Szemeredi{} regularity variants that we proved, although the algebra is a bit longer, due to the need to arrange and analyze a nested partition.

\paragraph{Sparse Graph Regularity Lemmas.}

Our first graph regularity lemma, Theorem \ref{thm:regdisc}, promises a vertex partition $V = V_1 \cup \dots \cup V_k$ satisfying
$$\sum \limits_{i, j=1}^k \disc(V_i, V_j) \le \eps n^2.$$ 
Any partition at all will give a bound of
$$\sum \limits_{i, j=1}^k \disc(V_i, V_j) \le |E|,$$
and therefore Theorem \ref{thm:regdisc} is nontrivial only when $|E| \ll \eps n^2$.
We might dream of an improved version of Theorem \ref{thm:regdisc} with a bound of the form
$$\sum \limits_{i, j=1}^k \disc(V_i, V_j) \overset{?}{\le} \eps |E|,$$ 
which would therefore be nontrivial on all graphs.
Sadly, this is too good to be true in general, but the area of \emph{sparse regularity lemmas} studies specific graph classes on which stronger versions of regularity lemmas with error bounds of the form $\eps |E|$ do exist.
At a technical level, this often boils down to graph classes on which the prominent application of Cauchy-Schwarz in Lemma \ref{lem:discsquare} loses at most a constant factor.

\paragraph{Other Inner Products.}

There is nothing very special about the Frobenius inner product in this discussion; all of these regularity lemmas can be extended to an arbitrary inner product space.
Let $V$ be an inner product space (for example, the space of $m\times n$ matrices equipped with the Frobenius inner product in Theorem \ref{thm:slcra}.)
Let $C\sbs V$ be a distinguished subset of ``atomic'' elements (for example, cut matrices in Theorem \ref{thm:slcra}).
For an element $A\in V$, define the $C$-rank of $A$ to be the minimum $k$ such that one can write 
\[A = \sum_{j = 1}^k c_jC_j\]
where $c_j\in \R$ and $C_j\in C.$ If there is no such finite $k$, we say $A$ has infinite $C$-rank.
We write 
\[\|A\|_{C[k]} := \sup \limits_{A' \text{ is a projection of $A$ of $C$-rank $\le k$}} \| A' \|,\]
where the norm on the right-hand side is the norm induced by the inner product, and the projection is defined with respect to the inner product.
We then have the following theorem:
\begin{mytheorem}
\begin{theorem} [Generalized Strong Regularity Lemma] \label{thm:gslcra}
For any $A\in V, \eps > 0$, and function $f : \nn \to \nn$, \textbf{there exists} an integer $k \le f^{(\eps^{-2})}$
and $\Ahat\in V$ satisfying:
\begin{itemize}
\item (simplicity) $\Ahat$ has $C$-rank $\le k$, and
\item (approximating) $\|A - \Ahat\|_{\blacksquare[f(k)]} \le \eps \|A\|$.
\end{itemize}
\end{theorem}
\end{mytheorem}
The proof is exactly the same as before.
This type of generality was first developed by Lov{\'a}sz and Szegedy \cite{LS07}.

\paragraph{Extension to Graphons.}

A \emph{graphon} is, intuitively speaking, a model of the adjacency matrix for a ``well-behaved'' infinite graph with edge weights in the interval $[0, 1]$.
It is officially defined as a symmetric measurable function $W\colon [0,1]^2\to [0,1]$; we can roughly imagine $W(i, j)$ as the weight of the ``edge'' between the points $i, j$.
An individual finite graph $G$ can be embedded into a graphon by mapping each node to a disjoint subinterval of $[0, 1]$, and in some sense graphons are the natural limit points of ``convergent'' infinite sequences of graphs.

Regularity lemmas extend to graphons as well, which are technically more general than our graph phrasings.
Given subsets $V_i, V_j\sbs [0,1]$, their \emph{discrepancy} in a graphon $W$ is defined as the quantity
$$ \disc(V_i, V_j) := \min_{c_{ij}} \max \limits_{S \subseteq V_i, T \subseteq V_j} \left| \int_{S\times T}W - c_{ij} m(S)m(T) \right|,$$
where $m$ denotes the Lesbegue measure.
%By taking $V = L^2([0,1]^2)$ (the space of square-integrable function $f\colon [0,1]^2\to \R$) and $C$ to be cut functions (e.g. function of the form $1_{S\times T}$ for some measurable $S, T\sbs [0,1]$)
We then have:
\begin{mytheorem}
\begin{theorem} [\Szemeredi{} Regularity Lemma, Graphon Phrasing, Discrepancy Version]
For any graphon $W$ and any $\eps > 0$, there exists an integer $k$ bounded by
$$k \le \tower(O(\eps^{-2}))$$
and a partition $[0,1] = V_1 \cup \dots \cup V_k$ such that
$$\sum \limits_{i, j=1}^k \disc(V_i, V_j) \le \eps.$$ 
\end{theorem}
\end{mytheorem}

For more on this direction, we recommend the book by Lov{\' a}sz \cite{Lovasz2012}.

\section*{Acknowledgments}

We are grateful to Santosh Vempala and Yuval Filmus for helpful technical discussions.

\bibliographystyle{plain}
\bibliography{refs}

\end{document}